\documentclass[12pt]{article}
\usepackage[latin1]{inputenc}
\usepackage{lmodern}

\usepackage{amssymb, amsmath, amsthm}
\usepackage[a4paper,top=25mm,bottom=25mm,left=25mm,right=25mm]{geometry}
\usepackage{etex}

\usepackage{authblk} 
\usepackage{pifont}
\usepackage{graphicx}
\usepackage[usenames,dvipsnames,svgnames,table]{xcolor}
\usepackage[figuresright]{rotating}
\usepackage{xtab} 
\usepackage{longtable} 
\usepackage{footnote}
\usepackage[stable]{footmisc}
\usepackage{chngpage} 
\usepackage{pdflscape} 

\usepackage{pgfplots}
\usepackage{setspace}

\makesavenoteenv{tabular}
\usepackage{tabularx}
\usepackage{booktabs}
\usepackage{threeparttable}
\usepackage[referable]{threeparttablex} 
\newcolumntype{R}{>{\raggedleft\arraybackslash}X}
\newcolumntype{L}{>{\raggedright\arraybackslash}X}
\newcolumntype{C}{>{\centering\arraybackslash}X}
\newcolumntype{A}{>{\columncolor{gray!25}}C}
\newcolumntype{a}{>{\columncolor{gray!25}}c}

\usepackage{dcolumn} 
\newcolumntype{.}{D{.}{.}{-1}}

\usepackage{tikz}
\usetikzlibrary{arrows}
\usepackage[semicolon]{natbib}
\PassOptionsToPackage{hyphens}{url}\usepackage[breaklinks]{hyperref} 
\hypersetup{
  colorlinks   = true,    
  urlcolor     = blue,    
  linkcolor    = blue,    
  citecolor    = red      
}
\usepackage{microtype}
\usepackage[justification=centerfirst]{caption}

\usepackage[labelformat=simple]{subcaption}

\DeclareCaptionLabelFormat{parenthesis}{(#2)}
\makeatletter
\renewcommand\p@subfigure{\arabic{figure}.}
\makeatother

\usepackage{enumitem}

\setlist[itemize]{leftmargin=3\parindent}
\setlist[enumerate]{leftmargin=2\parindent}

\theoremstyle{plain}

\newtheorem{lemma}{Lemma}
\newtheorem{proposition}{Proposition}
\newtheorem{theorem}{Theorem}

\theoremstyle{definition}

\newtheorem{definition}{Definition}
\newtheorem{example}{Example}

\theoremstyle{remark}

\newtheorem{remark}{Remark}

\def\keywords{\vspace{.5em} 
{\textit{Keywords}:\,\relax%
}}

\def\JEL{\vspace{.5em} 
{\textbf{\emph{JEL} classification number}:\,\relax%
}}

\def\AMS{\vspace{.5em} 
{\textbf{\emph{AMS} classification number}:\,}}

\author{L\'aszl\'o Csat\'o\thanks{~e-mail: laszlo.csato@uni-corvinus.hu} }
\affil{Department of Operations Research and Actuarial Sciences \\ Corvinus University of Budapest \\ MTA-BCE ''Lend\"ulet'' Strategic Interactions Research Group \\ Budapest, Hungary}
\title{Measuring centrality by a generalization of degree\thanks{~We are grateful to Herbert Hamers for drawing our attention to centrality measures, and to Dezs\H{o} Bednay, Pavel Chebotarev and Tam\'as Sebesty\'en for useful advices. \newline
The research was supported by OTKA grant K 111797 and MTA-SYLFF (The Ryoichi Sasakawa Young Leaders Fellowship Fund) grant 'Mathematical analysis of centrality measures', awarded to the author in 2015.}}
\date{\today}

\begin{document}

\maketitle

\begin{abstract}
Network analysis has emerged as a key technique in communication studies, economics, geography, history and sociology, among others. A fundamental issue is how to identify key nodes, for which purpose a number of centrality measures have been developed.
This paper proposes a new parametric family of centrality measures called generalized degree. It is based on the idea that a relationship to a more interconnected node contributes to centrality in a greater extent than a connection to a less central one. Generalized degree improves on degree by redistributing its sum over the network with the consideration of the global structure. Application of the measure is supported by a set of basic properties. A sufficient condition is given for generalized degree to be rank monotonic, excluding counter-intuitive changes in the centrality ranking after certain modifications of the network. The measure has a graph interpretation and can be calculated iteratively.
Generalized degree is recommended to apply besides degree since it preserves most favourable attributes of degree, but better reflects the role of the nodes in the network and has an increased ability to distinguish among their importance. 

\JEL{D85, Z13}

\AMS{15A06, 91D30}

\keywords{Network; centrality measure; degree; axiomatic approach}
\end{abstract}

\section{Introduction} \label{Sec1}

In recent years there has been a boom in network analysis. One fundamental concept that researchers try to capture is \emph{centrality}: a quantitative measure revealing the importance of nodes in the network. The first efforts to formally define centrality were made by \citet{Bavelas1948} and \citet{Leavitt1951}. Since then, a lot of centrality measures have been suggested (for a survey, see \citet{WassermanFaust1994}; for a short historical account, see \citet{BoldiVigna2014}; for some applications, see \citet{Jackson2010}).

Despite the agreement that centrality is an important attribute of networks, there is no consensus on its accurate meaning \citep{Freeman1979}. The central node of a star is obviously more important than the others but its role can be captured in several ways: it has the highest degree, it is the closest to other nodes \citep{Bavelas1948}, it acts as a bridge along the most shortest paths \citep{Freeman1977}, it has the largest number of incoming paths \citep{Katz1953, Bonacich1987}, or it maximizes the dominant eigenvector of a matrix adequately representing the network \citep{Seeley1949, BrinPage1998}.

An important question is the domain of centrality measures. We examine symmetric and unweighted networks, i.e. links have no direction and they are equally important, but non-connectedness is allowed.

The goal of current research is to introduce a centrality concept on the basis of degree by taking into account the whole structure of the network. It is achieved through the Laplacian matrix of the network.

This idea has been adopted earlier. For instance, \citet{ChebotarevShamis1997b_eng} introduce a connectivity index from which a number of centrality measures can be derived.
\citet{Klein2010} defines an edge-centrality measure exhibiting several nice features.
\citet{AvrachenkovMazalovTsynguev2015} propose a new concept of betweenness centrality for weighted networks with the Laplacian matrix.
\citet{MasudaKawamuraKori2009}, \citet{MasudaKori2010} and \citet{RanjanZhang2013} also use it in order to reveal the overall position of a node in a network.

We will consider a less known paired comparison-based ranking method called \emph{generalized row sum} \citep{Chebotarev1989_eng, Chebotarev1994} for the purpose.\footnote{~\citet[p.~1511]{ChebotarevShamis1997b_eng} note that 'there exists a certain relation between the problem of centrality evaluation and the problem of estimating the strength of players from incomplete tournaments'. The similarity between the two areas is mentioned by \citet{MonsuurStorcken2004}, too.}
In fact, our centrality measure redistributes the pool of aggregated degree (i.e. sum of degrees over the network) by considering all connections among the nodes, therefore it will be called \emph{generalized degree}. The impact of indirect connections is governed by a parameter such that one limit of generalized degree results in degree, while the other leads to equal centrality for all nodes (in a connected graph).

While there is a large literature in mathematical sociology on centrality measures, their comparison and evaluation clearly requires further investigation. We have chosen the axiomatic approach, a standard path in cooperative game and social choice theory, in order to confirm the validity of generalized degree for measuring the importance of nodes.
This line is followed by the following authors, among others.
\citet{Freeman1979} states that all centrality measures have an implicit starting point: the center of a star is the most central possible position.
\citet{Sabidussi1966} defines five properties that should be satisfied by a sensible centrality on an undirected graph. These axioms are also accepted by \citet{Nieminen1974}.
\citet{LandherrFriedlHeidemann2010} analyse five common centrality measures on the basis of three simple requirements concerning their behaviour.
\citet{BoldiVigna2014} introduce three axioms for directed networks, namely size, density and score monotonicity, and check whether they are satisfied by eleven standard centrality measures.

Though, characterizations of centrality measures are scarce.
\citet{Kitti2012} provides an axiomatization of eigenvector centrality on weighted networks.
\citet{Garg2009} characterizes some measures based on shortest paths and shows that degree, closeness and decay centrality belong to the same family, obtained by adding only one axiom to a set of four.
\citet{DequiedtZenou2015} present axiomatizations of Katz-Bonacich, degree and eigenvector centrality founded on the consistency axiom, which relates the properties of the measure for a given network to its behaviour on a reduced network. Similarly to our paper, \citet{Garg2009} and \citet{DequiedtZenou2015} use the domain of symmetric, unweighted networks.

Centrality measures are often used in order to identify the nodes with the highest importance, i.e. the \emph{center} of the network. \citet{MonsuurStorcken2004} present an axiomatic characterization of three different center concepts for connected, symmetric and unweighted networks.

However, a complete axiomatization of generalized degree will not be provided. While it is not debated that such characterizations are a correct way to distinguish between centrality measures, we think they have limited significance for applications. If one should determine the centrality of the nodes in a \emph{given} network, he/she is not much interested in the properties of the measure on smaller networks. Characterizations can provide some aspects of the choice but the consequences of the axioms on the actual network often remain obscure. From this viewpoint, the normative approach of \citet{Sabidussi1966}, \citet{LandherrFriedlHeidemann2010}, or \citet{BoldiVigna2014} seems to be more advantageous.

Our axiomatic scrutiny is mainly based on a property from \citet{Sabidussi1966} with some modification (in fact, strengthening) to eliminate the possibility of counter-intuitive changes in the centrality ranking of nodes. This requirement is well-known in paired comparison-based ranking (see, for instance, \citet{Gonzalez-DiazHendrickxLohmann2013}). In the case of centrality measures, it has been discussed by \citet{Chienetal2004} (with a proof that PageRank satisfies it) and proposed by \citet{BoldiVigna2014} as an essential counterpoint to score monotonicity. \citet{LandherrFriedlHeidemann2010} analyse similar properties of centrality measures, however, they mainly concentrate on the change of centrality scores (with the exception of Property 3, not discussed here). Given the aim of most applications, i.e. to distinguish the nodes with respect to their influence, it makes sense to take this relative point of view.

This requirement is called \emph{adding rank monotonicity}. A sufficient condition is given for the proposed measure to satisfy them, which is an important contribution of us.

It will also be presented that in a star network, generalized degree associates the highest value to its center, and it means a good tie-breaking rule of degree with an appropriate parameter choice.
On the basis of these results, the measure is recommended to use besides degree since the measure preserves most favourable attributes of degree but better reflects the role of the nodes in the network and has a much higher differential level.

The axiomatic point of view is not exclusive. \citet{BorgattiEverett2006} criticize \citet{Sabidussi1966}'s approach because it does not 'actually attempt to explain what centrality is'. Instead of this, \citet{BorgattiEverett2006} present a graph-theoretic review of centrality measures that classifies them according to the features of their calculation.
Thus a clear interpretation of generalized degree on the network graph will also be given, revealing that it is similar to degree-like walk-based measures: a node's centrality is a function of the centrality of the nodes it is connected to, and a relationship to a more interconnected node contributes to the own centrality to a greater extent than a connection to a less central one.

The paper proceeds as follows. Section~\ref{Sec2} defines the framework, introduces the centrality measure and presents its properties. In Section~\ref{Sec3}, we discuss the parameter choice by an axiomatic analysis, and highlight the differences to degree. Section~\ref{Sec4} gives an interpretation for generalized degree on the network graph through an iterative decomposition. Finally, Section~\ref{Sec5} summarizes the main results and draws the directions of future research.
Because of a new measure is presented, the paper contains more thoroughly investigated examples than usual.

\section{Generalized degree centrality} \label{Sec2}

We consider a finite set of nodes $N = \{ 1,2, \dots ,n \}$. A \emph{network} defined on $N$ is an unweighted, undirected graph (without loops or multiple edges) with the set of nodes $N$. The \emph{adjacency matrix} representation is adopted, the network is given by $(N,A)$ such that $A \in \mathbb{R}^{n \times n}$ is a symmetric matrix, $a_{ij} = 1$ if nodes $i$ and $j$ are connected and $a_{ij} = 0$ otherwise. If it does not cause inconvenience, the underlying graph will also be referred to as the network. Two nodes $i, j \in N$ are called \emph{symmetric} if a relabelling is possible such that the positions of $i$ and $j$ are interchanged and the network still has the same structure.

A \emph{path} between two nodes $i, j \in N$ is a sequence $(i = k_0, k_1, \dots ,k_m = j)$ of nodes such that $a_{k_\ell k_{\ell + 1}} = 1$ for all $\ell = 0,1, \dots ,m-1$. The network is called \emph{connected} if there exists a path between two arbitrary nodes.
The network graph should not be connected. A maximal connected subnetwork of $(N,A)$ is a \emph{component} of the network.
Let $\mathcal{N}$ denote the finite set of networks defined on $N$, and $\mathcal{N}^n$ denote the class of all networks $(N,A) \in \mathcal{N}$ with $|N| = n$.

Vectors are indicated by bold fonts and assumed to be column vectors.
Let $\mathbf{e} \in \mathbb{R}^{n}$ be given by $e_i = 1$ for all $i = 1,2, \dots ,n$ and $I \in \mathbb{R}^{n \times n}$ be the identity matrix, i.e., $I_{ii} = 1$ for all $i = 1,2, \dots ,n$ and $I_{ij} = 0$ for all $i \neq j$.

\begin{definition} \label{Def1}
\emph{Centrality measure}:
Let $(N,A) \in \mathcal{N}^n$ be a network. \emph{Centrality measure} $f$ is a function that assigns an $n$-dimensional vector of nonnegative real numbers to $(N,A)$ with $f_i(N,A)$ being the centrality of node $i$.
\end{definition}

A centrality measure will be denoted by $f: \mathcal{N}^n \to \mathbb{R}^n$. We focus on the centrality ranking, so centrality measures are invariant under multiplication by positive scalars (normalization): node $i$ is said to be at least as central as node $j$ in the network $(N,A)$ if and only if $f_i(N,A) \geq f_j(N,A)$.

\begin{definition} \label{Def2}
\emph{Degree}:
Let $(N,A) \in \mathcal{N}^n$ be a network. \emph{Degree} centrality $\mathbf{d}: \mathcal{N}^n \to \mathbb{R}^n$ is given by $\mathbf{d} = A \mathbf{e}$.
\end{definition}

A network is called \emph{regular} if all nodes have the same degree.

Degree is probably the oldest measure of importance ever used. It is usually a good baseline to approximate centrality. The major disadvantage of degree is that indirect connections are not considered at all, it does not reflect whether a given node is connected to central or peripheral nodes \citep{LandherrFriedlHeidemann2010}. This attribute is captured by the following property. 

\begin{definition} \label{Def3}
\emph{Independence of irrelevant connections} ($IIC$):
Let $(N,A),(N,A') \in \mathcal{N}^n$ be two networks and $k,\ell \in N$ be two distinct nodes such that $A$ and $A'$ are identical but $a'_{k \ell} = 1- a_{k \ell}$ (and $a'_{\ell k} = 1- a_{\ell k}$).
Centrality measure $f: \mathcal{N}^n \to \mathbb{R}^n$ is called \emph{independent of irrelevant connections} if $f_i(N,A) \geq f_j(N,A) \Rightarrow f_i(N,A') \geq f_j(N,A')$ for all $i,j \in N \setminus \{ k,\ell \}$.
\end{definition}

Independence of irrelevant connections is an adaptation of the axiom \emph{independence of irrelevant matches}, defined for general tournaments \citep{Rubinstein1980, Gonzalez-DiazHendrickxLohmann2013}. It is used in a modified form for a characterization of the degree center (i.e. nodes with the highest degree) under the name \emph{partial independence} \citep{MonsuurStorcken2004}.

\begin{lemma} \label{Lemma1}
Degree satisfies $IIC$.
\end{lemma}

Example~\ref{Examp1} shows that independence of irrelevant connections is an axiom one would rather not have.

\begin{figure}[htbp]
\centering
\caption{Networks of Example~\ref{Examp1}}
\label{Fig1}
  
\begin{subfigure}{\textwidth}
  \centering
  \subcaption{Network $(N,A)$}
  \label{Fig1a}
\begin{tikzpicture}[scale=1,auto=center, transform shape, >=triangle 45]
\tikzstyle{every node}=[draw,shape=circle];
  \node (n1) at (0,0) {$1$};
  \node (n2) at (2,0) {$2$};
  \node (n3) at (4,0) {$3$};
  \node (n4) at (6,0) {$4$};
  \node (n5) at (8,0) {$5$};

  \foreach \from/\to in {n1/n2,n2/n3,n3/n4}
    \draw (\from) -- (\to);
\end{tikzpicture}
\end{subfigure}

\vspace{0.5cm}
\begin{subfigure}{\textwidth}
  \centering
  \subcaption{Network $(N,A')$}
  \label{Fig1b}
\begin{tikzpicture}[scale=1,auto=center, transform shape, >=triangle 45]
\tikzstyle{every node}=[draw,shape=circle];
  \node (n1) at (0,0) {$1$};
  \node (n2) at (2,0) {$2$};
  \node (n3) at (4,0) {$3$};
  \node (n4) at (6,0) {$4$};
  \node (n5) at (8,0) {$5$};

  \foreach \from/\to in {n1/n2,n2/n3,n3/n4,n4/n5}
    \draw (\from) -- (\to);
\end{tikzpicture}
\end{subfigure}
\end{figure}
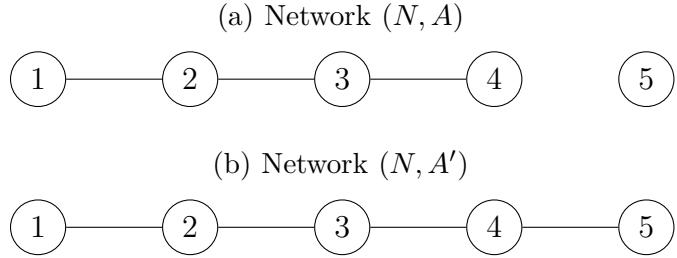

\begin{example} \label{Examp1}
Consider the networks $(N,A), (N,A') \in \mathcal{N}^5$ on Figure~\ref{Fig1}. Note that nodes $1$ and $4$, and $2$ and $3$ are symmetric in $(N,A)$, moreover, nodes $1$ and $5$, and $2$ and $4$ are symmetric in $(N,A')$.
$d_2 = d_3 = 2$ in both cases, but node $3$ seems to be more central in $(N,A')$ than node $2$.
\end{example}

The following centrality measure will be able to eliminate this shortcoming of degree.
It uses the \emph{Laplacian matrix} $L = \left( \ell_{ij} \right) \in \mathbb{R}^{n \times n}$ of the network $(N,A)$, given by $\ell_{ij} = -a_{ij}$ for all $i \neq j$ and $\ell_{ii} = d_i$ for all $i = 1,2, \dots ,n$.

\begin{definition} \label{Def4}
\emph{Generalized degree}:
Let $(N,A) \in \mathcal{N}^n$ be a network. \emph{Generalized degree} centrality $\mathbf{x}(\varepsilon): \mathcal{N}^n \to \mathbb{R}^n$ is given by $(I + \varepsilon L) \mathbf{x}(\varepsilon) = \mathbf{d}$, where $\varepsilon > 0$ is a parameter.
\end{definition}

Parameter $\varepsilon$ reflects the role of indirect connections, it is responsible for taking into account the centrality of neighbours, hence for breaking $IIC$. 

\begin{remark} \label{Rem1}
A straightforward formula for the generalized degree of node $i$ is
\begin{equation} \label{eq_main}
d_i = x_i(\varepsilon) + \varepsilon \sum_{j \in N \setminus \{ i \}} a_{ij} \left[ x_i(\varepsilon) - x_j(\varepsilon)\right]. 
\end{equation}
\end{remark}

\begin{example} \label{Examp2}
Consider the networks $(N,A), (N,A') \in \mathcal{N}^5$ on Figure~\ref{Fig1}. 
Generalized degree centrality is as follows:
\[
\mathbf{x}(\varepsilon)(N,A) = \left[ \frac{1+ 3 \varepsilon}{1 + 2 \varepsilon}, \, \frac{2+ 3 \varepsilon}{1 + 2 \varepsilon}, \, \frac{2+ 3 \varepsilon}{1 + 2 \varepsilon}, \, \frac{1+ 3 \varepsilon}{1 + 2 \varepsilon}, \, 0 \right]^\top;
\]
\[
\mathbf{x}(\varepsilon)(N,A') = \frac{1}{1 + 8 \varepsilon + 21 \varepsilon^2 + 20 \varepsilon^3 + 5 \varepsilon^4} \left[
\begin{array}{c}
1 + 9 \varepsilon + 27 \varepsilon^2 + 30 \varepsilon^3 + 8 \varepsilon^4 \\
2 + 15 \varepsilon + 37 \varepsilon^2 + 33 \varepsilon^3 + 8 \varepsilon^4 \\
2 + 16 \varepsilon + 40 \varepsilon^2 + 34 \varepsilon^3 + 8 \varepsilon^4 \\
2 + 15 \varepsilon + 37 \varepsilon^2 + 33 \varepsilon^3 + 8 \varepsilon^4 \\
1 + 9 \varepsilon + 27 \varepsilon^2 + 30 \varepsilon^3 + 8 \varepsilon^4 \\
\end{array}
\right].
\]
Thus $x_2(\varepsilon)(N,A) = x_3(\varepsilon)(N,A)$ but $x_2(\varepsilon)(N,A') < x_3(\varepsilon)(N,A')$.
\end{example}

\begin{remark} \label{Rem2}
Generalized degree violates $IIC$ for any $\varepsilon > 0$.
\end{remark}

Some basic attributes of generalized degree are listed below.

\begin{proposition} \label{Prop1}
Generalized degree satisfies the following properties for any fixed parameter $\varepsilon > 0$:
\begin{enumerate}
\item
\emph{Existence and uniqueness}: a unique vector $\mathbf{x}(\varepsilon)$ exists for any network $(N,A) \in \mathcal{N}$.

\item
\emph{Anonymity} ($ANO$): if the networks $(N,A),(\sigma N, \sigma A) \in \mathcal{N}$ are such that $(\sigma N, \sigma A)$ is given by a permutation of nodes $\sigma: N \rightarrow N$ from $(N,A)$, then $x_i(\varepsilon)(N,A) = x_{\sigma i}(\varepsilon)(\sigma N, \sigma A)$ for all $i \in N$.

\item
\emph{Degree preservation}: $\sum_{i \in N} x_i(\varepsilon) = \sum_{i \in N} d_i$ for any network $(N,A) \in \mathcal{N}$.

\item
\emph{Zero presumption} ($ZP$): $x_i(\varepsilon) = 0$ if and only if $d_i = 0$, $i$ is an isolated node.

\item
\emph{Independence of disconnected parts} ($IDCP$): if the networks $(N,A),(N,A') \in \mathcal{N}$ are such that $N^1 \cup N^2 = N$, $N^1 \cap N^2 = \emptyset$ and $a_{ik} = 0$, $a'_{ik} = 0$ for all $i \in N^1$, $k \in N^2$ and $a_{ij} = a'_{ij}$ for all $i,j \in N^1$, then $x_i(\varepsilon)(N,A) = x_i(\varepsilon)(N,A')$ for all $i \in N^1$ and $\sum_{i \in N^1} x_i(\varepsilon)(N,A) = \sum_{i \in N^1} x_i(\varepsilon)(N,A') = \sum_{i \in N^1} d_i$.\footnote{~Degrees in the component given by the set of nodes $N^1$ are the same in $(N,A)$ and $(N,A')$ due to the condition $a_{ij} = a'_{ij}$ for all $i,j \in N^1$.}

\item
\emph{Boundedness}: $\min \{ d_i: i \in N \} \leq x_j(\varepsilon) \leq \max \{ d_i: i \in N \}$ for all $j \in N$ and for any network $(N,A) \in \mathcal{N}$.

\item
\emph{Agreement}: $\lim_{\varepsilon \to 0} \mathbf{x}(\varepsilon) = \mathbf{d}$ and $\lim_{\varepsilon \to \infty} \mathbf{x}(\varepsilon) = \left( \sum_{i \in N} d_i / n \right) \mathbf{e}$ for any connected network $(N,A) \in \mathcal{N}$.

\item
\emph{Flatness preservation} ($FP$): $x_i(\varepsilon) = x_j(\varepsilon)$ for all $i,j \in N$ if and only if the network is regular.
\end{enumerate}

\end{proposition}

\begin{proof}
The statements above will be proved in the corresponding order.
\begin{enumerate}
\item
The Laplacian matrix of an undirected graph is positive semidefinite \citep[Theorem 2.1]{Mohar1991}, hence $I + \varepsilon L$ is positive definite.

\item
Generalized degree is invariant under isomorphism, it depends just on the structure of the graph and not on the particular labelling of the nodes.

\item
Sum of columns of $L$ is zero.

\item
If $d_i = 0$, then $x_i(\varepsilon) = 0$ since the corresponding row of $L$ contains only zeros. If $x_i(\varepsilon) = 0$ then $x_j(\varepsilon) \geq x_i(\varepsilon)$ for any $j \in N$, so $d_i = 0$ due to formula \eqref{eq_main}.

\item
Formula~\eqref{eq_main} gives the same equation for node $i \in N^1$ in the case of $(N,A)$ and $(N,A')$, and generalized degree is unique.
Sum of equations concerning nodes in $N^1$ gives $\sum_{i \in N^1} x_i(\varepsilon)(N,A) = \sum_{i \in N^1} x_i(\varepsilon)(N,A') = \sum_{i \in N^1} d_i$.

\item
Let $x_j(\varepsilon) = \min \{ x_i(\varepsilon): i \in N \}$.
Equation~\eqref{eq_main} results in
\[
x_j(\varepsilon) + \varepsilon \sum_{k \in N} a_{jk} \left[ x_j(\varepsilon) - x_k(\varepsilon) \right] = d_j,
\]
where the second term of the sum on the left-hand side is non-positive and $d_j \ \geq \min \{ d_i: i \in N \}$. The other inequality can be shown analogously.

\item
The first identity, $\lim_{\varepsilon \to 0} \mathbf{x}(\varepsilon) = \mathbf{d}$, is obvious. \\
Let $\lim_{\varepsilon \to \infty} x_j(\varepsilon) = \max \{ \lim_{\varepsilon \to \infty} x_i(\varepsilon): i \in N \}$. If $\lim_{\varepsilon \to \infty} x_j(\varepsilon) > \lim_{\varepsilon \to \infty} x_h(\varepsilon)$ for any $h \in N$, then exists $k,m \in N$ such that $\lim_{\varepsilon \to \infty} x_k(\varepsilon) = \lim_{\varepsilon \to \infty} x_j(\varepsilon) = \max \{ \lim_{\varepsilon \to \infty} x_i(\varepsilon): i \in N \} > \lim_{\varepsilon \to \infty} x_m(\varepsilon)$ and $a_{k m} = 1$ due to the connectedness of $(N,A)$.\footnote{~There exists a path from node $j$ to node $h$ due to connectedness. Along this path one should find such $k,m \in N$, for example, $k = j$ and $m = h$ if $a_{jh} = 1$.}
But $d_k = x_k(\varepsilon) + \varepsilon \sum_{\ell \in N} a_{k \ell} \left[ x_k(\varepsilon) - x_\ell(\varepsilon) \right] \geq \varepsilon \left[ x_k(\varepsilon) - x_m(\varepsilon) \right]$ from formula~\eqref{eq_main}, which is impossible when $\varepsilon \to \infty$.

\item
It can be verified that $x_i(\varepsilon) = d_i$ satisfies $(I + \varepsilon L) \mathbf{x}(\varepsilon) = \mathbf{d}$ if $d_i = d_j$ for all $i,j \in N$. $x_i(\varepsilon) = x_j(\varepsilon)$ for all $i,j \in N$ implies $L \mathbf{x}(\varepsilon) = \mathbf{0}$, so $\mathbf{x}(\varepsilon) = \mathbf{d}$.
\end{enumerate}
\end{proof}

$ANO$ contains \emph{symmetry} \citep{Garg2009}, namely, two symmetric nodes have equal centrality. It provides that all nodes have the same centrality in a complete network, too.

The name of the measure comes from the property degree preservation, it can be perceived as a centrality measure redistributing the sum of degree among the nodes.

$ZP$ and $INRP$ address the issue of disconnected networks. $ZP$ is an extension of the axiom \emph{isolation} \citep{Garg2009}, demanding that an isolated node has zero centrality. $INRP$ shows that the centrality in a component of the network are independent from other components.

Boundedness means that centrality is placed on an interval not broader than in the case of degree. It is easy to prove that the stronger condition of $\min \{ d_i: i \in N \} < x_j(\varepsilon)$ or $x_j(\varepsilon) < \max \{ d_i: i \in N \}$ is also satisfied if node $j$ is at least indirectly connected to a node with a greater or smaller degree, respectively.

According to agreement, the limits of generalized degree are degree and equal centrality on all connected components of a network (because of independence of disconnected parts).

$FP$ means that generalized degree results in a tied centrality between any nodes if and only if degree also gives equal centrality for all nodes. Note that it is true for any fixed $\varepsilon$, so it could not occur that generalized degrees are tied between any two nodes only for certain parameter values.
It also shows that two nodes may have the same generalized degree for any $\varepsilon > 0$ not only if they are symmetric as there exists regular graphs with non-symmetric pairs of nodes.

Degree also satisfies the properties listed in Proposition~\ref{Prop1} (except for agreement).

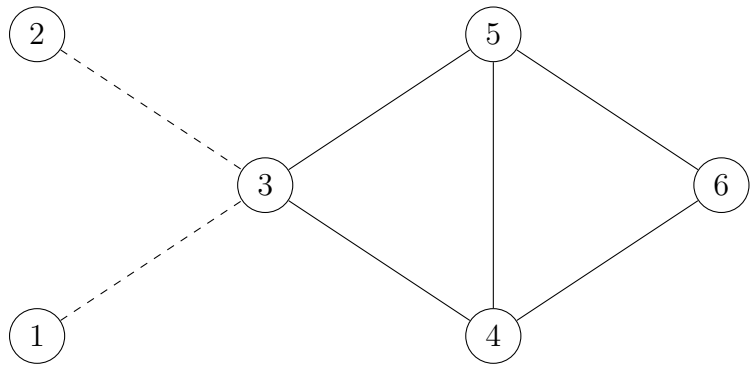
\begin{figure}[htbp]
\centering
\caption{Network of Example~\ref{Examp3}}
\label{Fig3}

\begin{tikzpicture}[scale=1,auto=center, transform shape, >=triangle 45]
\tikzstyle{every node}=[draw,shape=circle];
  \node (n1) at (0,-2) {$1$};
  \node (n2) at (0,2)  {$2$};
  \node (n3) at (3,0)  {$3$};
  \node (n4) at (6,-2) {$4$};
  \node (n5) at (6,2)  {$5$};
  \node (n6) at (9,0)  {$6$};

  \foreach \from/\to in {n3/n4,n3/n5,n4/n5,n4/n6,n5/n6}
    \draw (\from) -- (\to);
    
    \draw[dashed] (n1) -- (n3);
    \draw[dashed] (n2) -- (n3);
\end{tikzpicture}
\end{figure}

\begin{example} \label{Examp3}
Consider the network on Figure~\ref{Fig3} where both normal and dashed lines indicate connections. Generalized degrees with various values of $\varepsilon$ are given on Figure~\ref{Fig4}. Nodes $1$ and $2$, and $4$ and $5$ are symmetric thus the two pairs have the same centrality for any parameter values. Degree gives the ranking $3 \succ (4 \sim 5) \succ 6 \succ (1 \sim 2)$. Generalized degrees of nodes $3$, $4$ and $5$ monotonically decrease, while the generalized degree of nodes $1$ and $2$ increases. However, the centrality of node $6$ is \emph{not monotonic}, for certain $\varepsilon$-s it becomes larger than its limit of $7/3$ (see the property agreement in Proposition~\ref{Prop1}).

It results in two changes in the centrality ranking with the watersheds of $\varepsilon_1 = 1/2$ and $\varepsilon_2 = \left( 2 + \sqrt{6} \right)/2$. In the case of $0 < \varepsilon < \varepsilon_1$ the central node is $3$, while for $\varepsilon > \varepsilon_1$ the suggestion is $4$ and $5$. If $\varepsilon > \varepsilon_2$, then node $6$ becomes more central than node $3$.
\end{example}

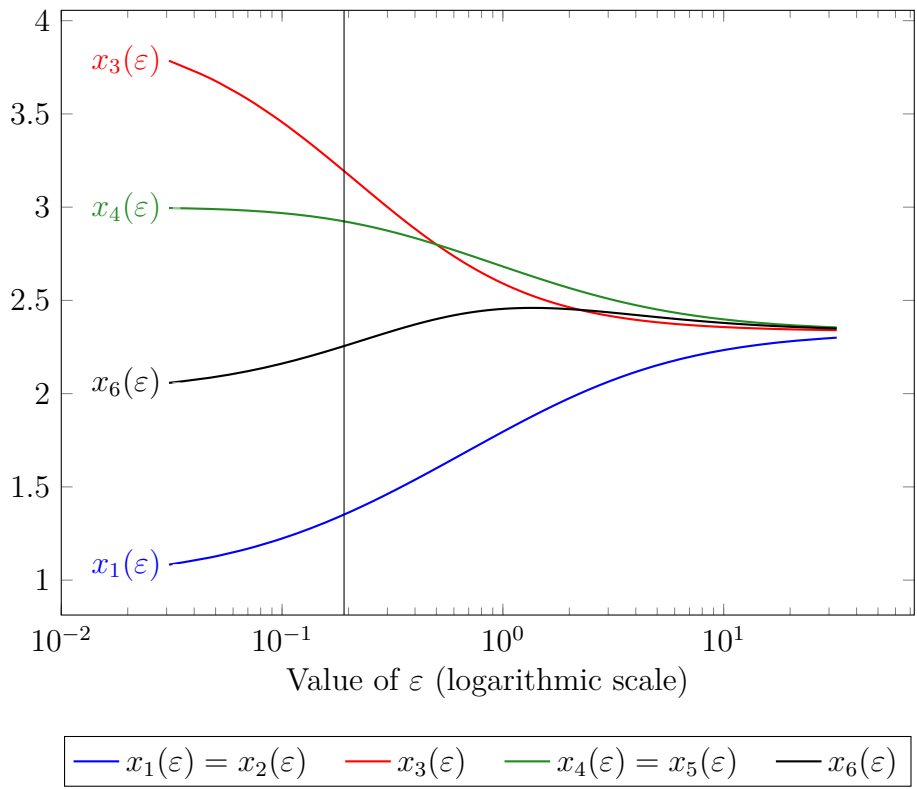
\begin{figure}[htbp]
\centering
\caption{Generalized degrees in Example~\ref{Examp3}}
\label{Fig4}

\begin{tikzpicture}
\begin{axis}[width=0.8\textwidth, 
height=0.6\textwidth,
xmode=log,
xlabel = Value of $\varepsilon$ (logarithmic scale),
xmin=0.01,
legend entries={$x_1(\varepsilon) = x_2(\varepsilon)\quad$,$x_3(\varepsilon)\quad$,$x_4(\varepsilon) = x_5(\varepsilon)\quad$,$x_6(\varepsilon)$},
legend style={at={(0.5,-0.2)},anchor = north,legend columns = 6}
]

\addplot[blue,smooth,thick] coordinates{
(0.0307554190699851,1.08310445176305)
(0.0471132987306222,1.12110904519587)
(0.0693009497745839,1.16734549785650)
(0.0982620958595216,1.22031189499391)
(0.1347726459997020,1.27786656975896)
(0.1793740787340170,1.33768944716062)
(0.2323284650095870,1.39767098705267)
(0.2935988413830830,1.45612938890368)
(0.3628547336948770,1.51186505051214)
(0.4394994556648180,1.56411413422738)
(0.5227137589848350,1.61246183557201)
(0.6115095555042990,1.65675267055398)
(0.7047875933626770,1.69701362435477)
(0.8013938473712240,1.73339364674125)
(0.9001706540012430,1.76611797430205)
(1.0000000000000000,1.79545454545455)
(1.1109004671003400,1.82442892795514)
(1.2478259014344200,1.85591165097083)
(1.4188672011503600,1.88989757007095)
(1.6352974225812700,1.92627935663849)
(1.9130929362603800,1.96481128225796)
(2.2753156735708100,2.00507191390655)
(2.7559238095565100,2.04643365191882)
(3.4060079913435900,2.08805149780877)
(4.3042508801439100,2.12888626744885)
(5.5749415247608800,2.16777520154949)
(7.4199032940424100,2.20355235045033)
(10.1768641433176000,2.23520258616743)
(14.4298166656115000,2.26201350870253)
(21.2254294847334000,2.28368066708303)
(32.5145951588065000,2.30033224704575)
}
node [pos=0,pin={[pin edge={white}, pin distance=-0.2cm] 180:{$x_1(\varepsilon)$}}] {};

\addplot[red,smooth,thick] coordinates{
(0.0307554190699851,3.78521205673314)
(0.0471132987306222,3.69170054396421)
(0.0693009497745839,3.58211035174123)
(0.0982620958595216,3.46239611950373)
(0.1347726459997020,3.33960964601770)
(0.1793740787340170,3.22028836860993)
(0.2323284650095870,3.10934668308183)
(0.2935988413830830,3.00970973259626)
(0.3628547336948770,2.92252613049838)
(0.4394994556648180,2.84765186551776)
(0.5227137589848350,2.78415824693390)
(0.6115095555042990,2.73073861998426)
(0.7047875933626770,2.68598339470670)
(0.8013938473712240,2.64854123509243)
(0.9001706540012430,2.61719878980816)
(1.0000000000000000,2.59090909090909)
(1.1109004671003400,2.56655565521007)
(1.2478259014344200,2.54183398195220)
(1.4188672011503600,2.51708633682055)
(1.6352974225812700,2.49270803428931)
(1.9130929362603800,2.46913141431800)
(2.2753156735708100,2.44680047297248)
(2.7559238095565100,2.42613705601519)
(3.4060079913435900,2.40750215693055)
(4.3042508801439100,2.39115868113565)
(5.5749415247608800,2.37724380249587)
(7.4199032940424100,2.36575828531968)
(10.1768641433176000,2.35657618109534)
(14.4298166656115000,2.34947224348397)
(21.2254294847334000,2.34415909782604)
(32.5145951588065000,2.34032451023386)
}
node [pos=0,pin={[pin edge={white}, pin distance=-0.2cm] 180:{$x_3(\varepsilon)$}}] {};

\addplot[ForestGreen,smooth,thick] coordinates{
(0.0307554190699851,2.99544815154589)
(0.0471132987306222,2.99039786059784)
(0.0693009497745839,2.98183970021050)
(0.0982620958595216,2.96891951664815)
(0.1347726459997020,2.95133644084291)
(0.1793740787340170,2.92946391447147)
(0.2323284650095870,2.90422471748691)
(0.2935988413830830,2.87682179402539)
(0.3628547336948770,2.84846926491738)
(0.4394994556648180,2.82021171090939)
(0.5227137589848350,2.79284552359519)
(0.6115095555042990,2.76691463775371)
(0.7047875933626770,2.74274563354983)
(0.8013938473712240,2.72049619765748)
(0.9001706540012430,2.70020235015975)
(1.0000000000000000,2.68181818181818)
(1.1109004671003400,2.66351011580303)
(1.2478259014344200,2.64347367527892)
(1.4188672011503600,2.62170552765171)
(1.6352974225812700,2.59827499195664)
(1.9130929362603800,2.57334850991122)
(2.2753156735708100,2.54721385870469)
(2.7559238095565100,2.52029865040563)
(3.4060079913435900,2.49317505523721)
(4.3042508801439100,2.46654127779456)
(5.5749415247608800,2.44117220447363)
(7.4199032940424100,2.41783868014376)
(10.1768641433176000,2.39720664160774)
(14.4298166656115000,2.37973940768748)
(21.2254294847334000,2.36563146503179)
(32.5145951588065000,2.35479476581800)
}
node [pos=0,pin={[pin edge={white}, pin distance=-0.2cm] 180:{$x_4(\varepsilon)$}}] {};

\addplot[black,smooth,thick] coordinates{
(0.0307554190699851,2.05768273664899)
(0.0471132987306222,2.08528564444837)
(0.0693009497745839,2.11951925212477)
(0.0982620958595216,2.15914105721215)
(0.1347726459997020,2.20198433277857)
(0.1793740787340170,2.24540490812589)
(0.2323284650095870,2.28686190783901)
(0.2935988413830830,2.32438790154561)
(0.3628547336948770,2.35680523864258)
(0.4394994556648180,2.38369644420869)
(0.5227137589848350,2.40522703473170)
(0.6115095555042990,2.42192676340037)
(0.7047875933626770,2.43449808948409)
(0.8013938473712240,2.44367907611012)
(0.9001706540012430,2.45016056126824)
(1.0000000000000000,2.45454545454545)
(1.1109004671003400,2.45756625727359)
(1.2478259014344200,2.45939536554829)
(1.4188672011503600,2.45970746773412)
(1.6352974225812700,2.45818326852043)
(1.9130929362603800,2.45454900134365)
(2.2753156735708100,2.44862798180503)
(2.7559238095565100,2.44039833933592)
(3.4060079913435900,2.43004473697747)
(4.3042508801439100,2.41798622837752)
(5.5749415247608800,2.40486138545788)
(7.4199032940424100,2.39145965349213)
(10.1768641433176000,2.37860536335432)
(14.4298166656115000,2.36702192373598)
(21.2254294847334000,2.35721663794429)
(32.5145951588065000,2.34942146403872)
}
node [pos=0,pin={[pin edge={white}, pin distance=-0.2cm] 180:{$x_6(\varepsilon)$}}] {};

\draw ({axis cs:0.190892435325843,0}|-{rel axis cs:0,0}) -- ({axis cs:0.190892435325843,0}|-{rel axis cs:0,1});
\end{axis}
\end{tikzpicture}
\end{figure}

\begin{figure}[htbp]
\centering
\caption{Network of Remark~\ref{Rem3}}
\label{Fig2}

\begin{tikzpicture}[scale=1,auto=center, transform shape, >=triangle 45]
\tikzstyle{every node}=[draw,shape=circle];
  \node (n1) at (112.5:3) {$1$};
  \node (n2) at (67.5:3)  {$2$};
  \node (n3) at (22.5:3)  {$3$};
  \node (n4) at (337.5:3) {$4$};
  \node (n5) at (292.5:3) {$5$};
  \node (n6) at (247.5:3) {$6$};
  \node (n7) at (202.5:3) {$7$};
  \node (n8) at (157.5:3) {$8$};

  \foreach \from/\to in {n1/n2,n1/n8,n2/n3,n3/n4,n4/n5,n5/n6,n6/n7,n7/n8}
    \draw (\from) -- (\to);
  
  \foreach \from/\to in {n1/n4,n2/n5,n3/n7,n6/n8}
    \draw (\from) -- (\to);
\end{tikzpicture}
\end{figure}
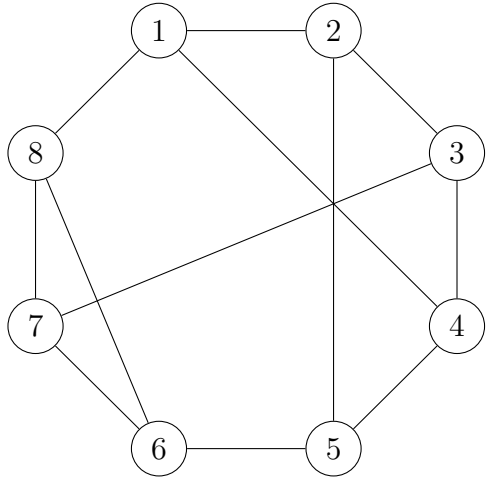

\begin{remark} \label{Rem3}
According to \citet{ChebotarevShamis1997b_eng}, the diagonal entries $q_{ii}$ of the matrix $Q = (I + \alpha L)^{-1}$, $\alpha > 0$ may serve as a measure of \emph{solitariness}, that is, an 'inverse' centrality index.

Despite its similar formula, this measure is not concordant with generalized degree. Consider the network on Figure~\ref{Fig2}, which is a $3$-regular graph with $8$ vertices. According to the property $FP$ of Proposition~\ref{Prop1}, generalized degree gives the ranking $(1 \sim 2 \sim \dots \sim 8)$, while $1-q_{ii}$  results in $(1 \sim 3 \sim 5) \succ (2 \sim 4) \succ (6 \sim 7 \sim 8)$. The latter method 'punishes' nodes $2$ and $4$ because they have a difficulty in reaching $6$, $7$ and $8$, and vice versa. However, it is not clear whether this differentiation is desirable in a regular graph.
\end{remark}

\section{An axiomatic review of generalized degree} \label{Sec3}

The interpretation of generalized degree on the network gives few information about the appropriate value of parameter $\varepsilon$. Here we use the axiomatic approach in order to get an insight into it.

\citet{Freeman1979} states that all centrality measures have an implicit starting point: the central node of a star is the most central possible position.

\begin{definition} \label{Def5}
\emph{Star center base} ($SCB$):
Let $(N,A) \in \mathcal{N}^n$ be a star network with $i \in N$ at the center, that is, $a_{ij} = 1$ for all $j \neq i$ and $a_{jk} = 0$ for all $j,k \in N \setminus \{ i \}$.
Centrality measure $f: \mathcal{N}^n \to \mathbb{R}^n$ is \emph{star center based} if $f_i(N,A) > f_j(N,A)$ for all $j \in N \setminus \{ i \}$.
\end{definition}

The following requirement capture the essence of 'dynamic' monotonicity with respect to the centrality ranking.

\begin{definition} \label{Def6}
\emph{Adding rank monotonicity} ($ARM$):
Let $(N,A),(N,A') \in \mathcal{N}^n$ be two networks and $i,j \in N$ be two distinct nodes such that $A$ and $A'$ are identical but $a_{ij} = 0$ and $a'_{ij} = 1$.
Centrality measure $f: \mathcal{N}^n \to \mathbb{R}^n$ is \emph{adding rank monotonic} if $f_i(N,A) \geq f_k(N,A) \Rightarrow f_i(N,A') \geq f_k(N,A')$ and $f_i(N,A) > f_k(N,A) \Rightarrow f_i(N,A') > f_k(N,A')$ for all $k \in N \setminus \{ i,j \}$.
\end{definition}



$ARM$ implies that, other nodes cannot benefit more from adding an edge between $i$ and $j$ than the nodes involved with respect to the centrality ranking. The supplementary condition excludes the possibility that node $i$ is more central than another node in the network $(N,A)$ but they are tied in $(N,A')$.

A version of $ARM$, focusing on the nodes with the highest centrality, appeared in \citet{Sabidussi1966} and \citet{Nieminen1974}.
According to our knowledge, in this form $ARM$ was first introduced by \citet{Chienetal2004} for directed graphs. \citet{BoldiVigna2014} argue for the use of rank monotonicity in order to exclude pathological, counter-intuitive changes in centrality. However, they leave the study of such an axiom for future work.
A review of other monotonicity axioms can be found in \citet{LandherrFriedlHeidemann2010} and \citet{BoldiVigna2014}.

Note that the centrality measure associating a constant value to every node of every network satisfies $ARM$ but does not meet $SCB$. Degree is star center based and adding rank monotonic.

\begin{proposition} \label{Prop2}
Generalized degree satisfies $SCB$.
\end{proposition}

\begin{proof}
Let $x_\ell = x_\ell(\varepsilon)(N,A)$ for all $\ell \in N$.
Due to anonymity (see Proposition~\ref{Prop1}), $x_j = x_k$ for any $j,k \in N \setminus \{ i \}$. Formula \eqref{eq_main} gives two conditions for the two variables:
\begin{eqnarray*}
\left[ 1 + \varepsilon (n-1) \right] x_i - \varepsilon (n-1) x_j & = & n-1 \\
\left( 1 + \varepsilon \right) x_i - \varepsilon x_j & = & 1. \\
\end{eqnarray*}
It can be checked that the solution is
\[
x_i = \frac{(n-1)(1 + 2 \varepsilon)}{1 + \varepsilon n} \qquad \text{and} \qquad
x_j = \frac{1 + \varepsilon (2n-1) + 2 \varepsilon^2 (n-1)}{(1 + \varepsilon) (1 + \varepsilon n)},
\]
hence
\[
x_i - x_j = \frac{n - 2}{1 + \varepsilon n} > 0.\footnote{~$n \geq 3$ is necessary to get a meaningful start network.}
\]
\end{proof}

$SCB$ means only a validity test, it is a 'natural' requirement for any centrality measure.

Now a sufficient condition is given for generalized degree to be adding rank monotonic. Let $\mathfrak{d} = \max \{ d_i: i \in N \}$ be the \emph{maximal} and $\mathsf{d} = \min \{ d_k: k \in N \}$ be the \emph{minimal degree}, respectively, in a network $(N,A) \in \mathcal{N}$.

\begin{theorem} \label{Theo1}
Generalized degree satisfies $ARM$ if
\[
(\mathfrak{d} - \mathsf{d}) \left[ (2\mathfrak{d} + 4) \varepsilon^3 + 2 \varepsilon^2 + \varepsilon \right] \leq 1.
\]
\end{theorem}

The proof of this statement is not elegant, therefore one can see it in the Appendix.

In a certain sense, Theorem~\ref{Theo1} is not surprising because degree satisfies $ARM$ and $\mathbf{x}(\varepsilon)$ is close to it when $\varepsilon$ is small.\footnote{~There exists an appropriately small $\varepsilon$ satisfying the condition of Theorem~\ref{Theo1} for any $\mathsf{d}$ and $\mathfrak{d}$.} Our main contribution is the calculation of a sufficient condition. Note that it depends only on the minimal and maximal degree as well as on $\varepsilon$ but not on the the number of nodes. Nevertheless, since the network graph is unweighted, $\mathfrak{d} - \mathsf{d} \leq n-1$, therefore the condition $(n-1) \left[ (2\mathfrak{d} + 4) \varepsilon^3 + 2 \varepsilon^2 + \varepsilon \right] \leq 1$ provides $ARM$, too.

Parameter $\varepsilon$ is called \emph{reasonable} if it satisfies $(\mathfrak{d} - \mathsf{d}) \left[ (2\mathfrak{d} + 4) \varepsilon^3 + 2 \varepsilon^2 + \varepsilon \right] \leq 1$ for the network $(N,A) \in \mathcal{N}$.
It is always satisfied for a regular network where $\mathfrak{d} = \mathsf{d}$. The value of $\varepsilon$ is decreasing in the maximal degree and, especially, in the difference of maximal and minimal degree. However, it does not become extremely small for sparse networks, where the difference of maximal and minimal degrees can be significantly smaller than the number of nodes.


Note the analogy to the (dynamic) monotonicity of generalized row sum method \citep[Property~13]{Chebotarev1994}.


Similarly to \citet{Sabidussi1966}, strict inequalities can be demanded in $ARM$ (i.e. adding an edge between nodes $i$ and $j$ eliminates all ties -- except for against $j$ -- with $i$ in the centrality ranking) but it does not affect our discussion, all results remain valid with a corresponding modification of inequalities.

According to the following example, violation of $ARM$ can be a problem in practice.

\begin{example} \label{Examp4}
Consider the networks $(N,A),(N,A') \in \mathcal{N}^6$ on Figure~\ref{Fig3} such that $(N,A)$ is given by the normal edges and $(N,A')$ is obtained from $(N,A)$ by adding the two dashed edges between nodes $1$ and $3$, and nodes $2$ and $3$.

Axiom $ARM$ demands that $x_i(\varepsilon)(N,A) \geq x_k(\varepsilon)(N,A) \Rightarrow x_i(\varepsilon)(N,A') \geq x_k(\varepsilon)(N,A')$ and $x_i(\varepsilon)(N,A) > x_k(\varepsilon)(N,A) \Rightarrow x_i(\varepsilon)(N,A') > x_k(\varepsilon)(N,A')$ for all $i = 1,2,3$ and $j = 4,5,6$. However, $x_3(\varepsilon=3)(N,A) = x_6(\varepsilon=3)(N,A)$ since nodes $3$ and $6$ are symmetric in $(N,A)$ but $x_3(\varepsilon=3)(N,A') < x_6(\varepsilon=3)(N,A')$ as can seen on Figure~\ref{Fig4}.
It is difficult to argue for this ranking since nodes $1$ and $2$ are only connected to node $3$.\footnote{~Nevertheless, the lower rank of node $3$ may be explained. \citet{LandherrFriedlHeidemann2010} mention cannibalization and saturation effects, which sometimes arise when an actor can devote less time to maintaining existing relationships as a result of adding new contacts. Similarly, the edges can represent not only opportunities, but liabilities, too. For example, a service provider may have legal constraints to serve unprofitable customers. Investigation of these models is leaved for future research.}

The root of the problem is the excessive influence of neighbours' degrees: the low values of $d_1$ and $d_2$ decrease generalized degree of node $3$ despite its degree becomes greater. For large $\varepsilon$-s this effect is responsible for breaking of the property $ARM$. 

Theorem~\ref{Theo1} gives the condition of reasonableness as
\[
30 \varepsilon^3 + 6 \varepsilon^2 + 3 \varepsilon \leq 1,
\]
because $\mathfrak{d} = d_3 = d_4 = d_5 = 3$ and $\mathsf{d} = d_1 = 0$ before connecting nodes $1$ and $3$ ($ARM$ allows for adding only one edge). It is satisfied if
\[
\varepsilon \leq \frac{1}{15} \left( \sqrt[3]{\frac{266 + 15 \sqrt{334}}{4}} - \frac{13}{\sqrt[3]{532 + 30 \sqrt{334}}} - 1 \right) \approx 0.1909.
\]
This upper bound of $\varepsilon$ is indicated by the vertical line on Figure~\ref{Fig4}.
Note that $ARM$ is violated only if $\varepsilon > \left(2 + \sqrt{6} \right) / 2$ according to Example~\ref{Examp3}, Theorem~\ref{Theo1} does not give a necessary condition for $ARM$.
\end{example}

It is worth to scrutinize the connection of degree and generalized degree further. Example~\ref{Examp3} verifies that generalized degree (with an appropriate value of $\varepsilon$) is not only a tie-breaking rule of degree, there exists networks where a node with a smaller degree has a larger centrality by $\mathbf{x}(\varepsilon)$.

A special type of networks is provided by a given degree sequence. Then degree of nodes is fixed along with the reasonableness of $\varepsilon$ (depending only on degrees), while generalized degree may result in different rankings.

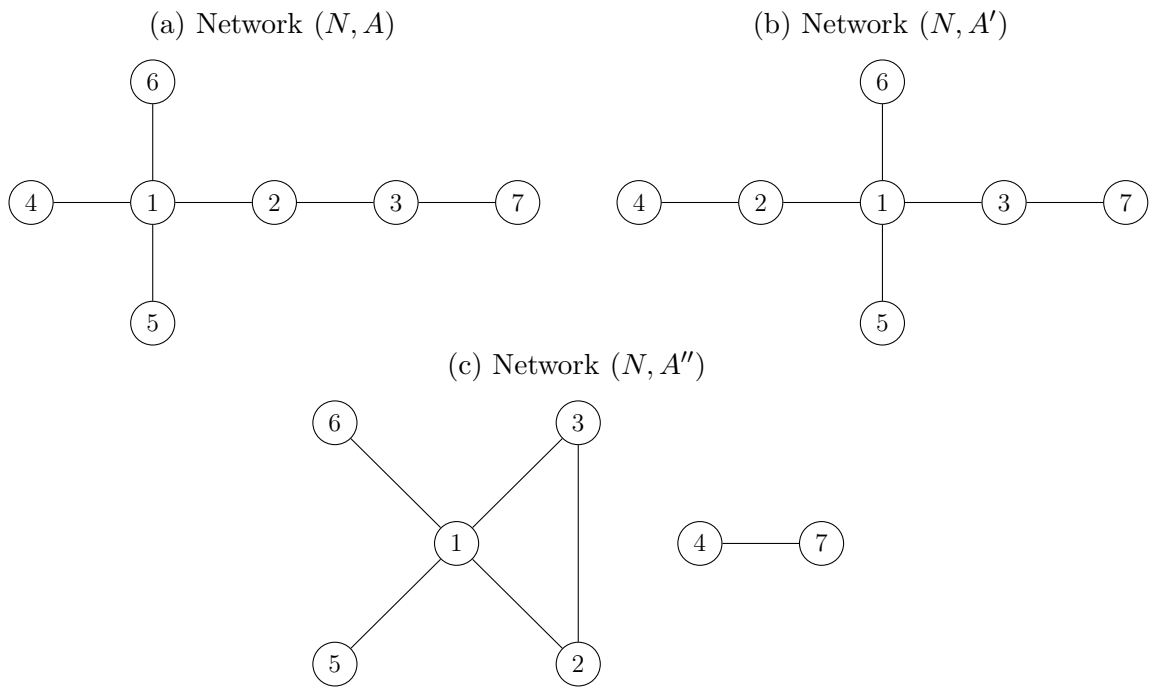
\begin{figure}[htbp]
\centering
\caption{Networks of Example~\ref{Examp5}}
\label{Fig5}
  
\begin{subfigure}{.5\textwidth}
  \centering
  \subcaption{Network $(N,A)$}
  \label{Fig5a}
\begin{tikzpicture}[scale=0.8, auto=center, transform shape, >=triangle 45]
\tikzstyle{every node}=[draw,shape=circle]; 
  \node (n1) at (2,0)  {$1$};
  \node (n2) at (4,0)  {$2$};
  \node (n3) at (6,0)  {$3$};
  \node (n4) at (0,0)  {$4$};
  \node (n5) at (2,-2) {$5$};
  \node (n6) at (2,2)  {$6$};
  \node (n7) at (8,0)  {$7$};

  \foreach \from/\to in {n1/n2,n1/n4,n1/n5,n1/n6,n2/n3,n3/n7}
    \draw (\from) -- (\to);
\end{tikzpicture}
\end{subfigure}
\begin{subfigure}{.5\textwidth}
  \centering
  \subcaption{Network $(N,A')$}
  \label{Fig5b}
\begin{tikzpicture}[scale=0.8, auto=center, transform shape, >=triangle 45]
\tikzstyle{every node}=[draw,shape=circle];
  \node (n1) at (4,0)  {$1$};
  \node (n2) at (2,0)  {$2$};
  \node (n3) at (6,0)  {$3$};
  \node (n4) at (0,0)  {$4$};
  \node (n5) at (4,-2) {$5$};
  \node (n6) at (4,2)  {$6$};
  \node (n7) at (8,0)  {$7$};

  \foreach \from/\to in {n1/n2,n1/n3,n1/n5,n1/n6,n2/n4,n3/n7}
    \draw (\from) -- (\to);
\end{tikzpicture}
\end{subfigure}
\begin{subfigure}{\textwidth}
  \centering
  \subcaption{Network $(N,A'')$}
  \label{Fig5c}
\begin{tikzpicture}[scale=0.8, auto=center, transform shape, >=triangle 45]
\tikzstyle{every node}=[draw,shape=circle];
  \node (n1) at (2,0)  {$1$};
  \node (n2) at (4,-2) {$2$};
  \node (n3) at (4,2)  {$3$};
  \node (n4) at (6,0)  {$4$};
  \node (n5) at (0,-2) {$5$};
  \node (n6) at (0,2)  {$6$};
  \node (n7) at (8,0)  {$7$};

  \foreach \from/\to in {n1/n2,n1/n3,n1/n5,n1/n6,n2/n3,n4/n7}
    \draw (\from) -- (\to);
\end{tikzpicture}
\end{subfigure}
\end{figure}

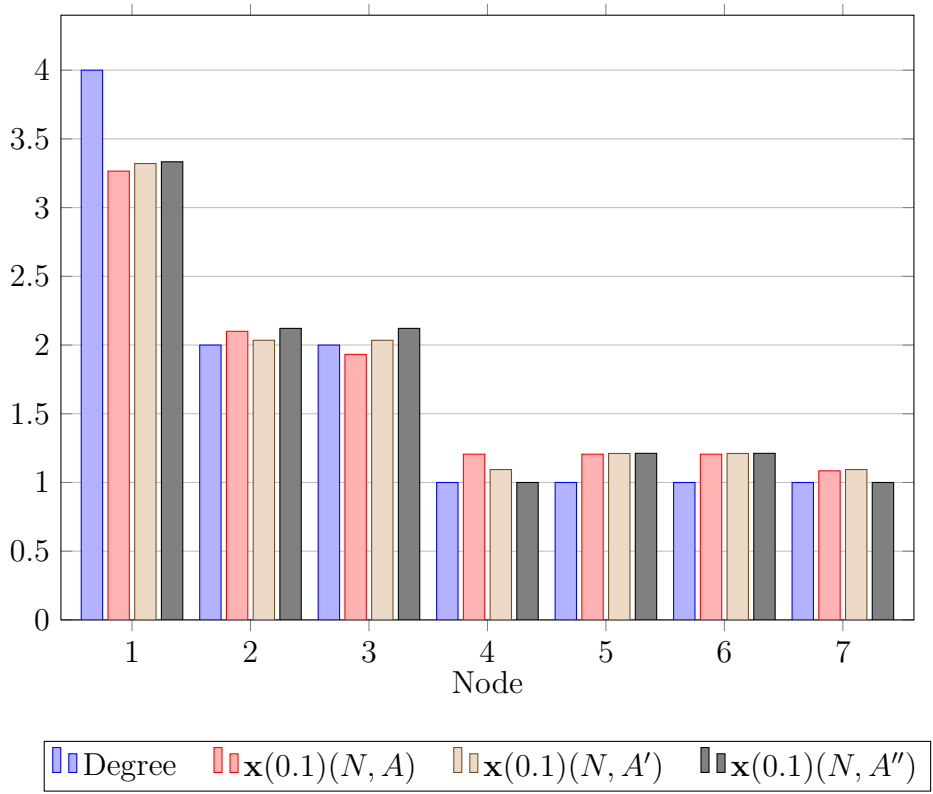
\begin{figure}[htbp]
\centering
\caption{Centralities in Example~\ref{Examp5}}
\label{Fig6}

\begin{tikzpicture}
\begin{axis}[width=0.8\textwidth, 
height=0.6\textwidth,
symbolic x coords={1,2,3,4,5,6,7},
xtick=data,
xlabel = Node,
ybar,
ymin = 0,
ymajorgrids = true,
bar width=8pt,
legend entries={Degree$\quad$,{$\mathbf{x}(0.1)(N,A)\quad$},{$\mathbf{x}(0.1)(N,A')\quad$},{$\mathbf{x}(0.1)(N,A'')$}},
legend style={at={(0.5,-0.2)},anchor = north,legend columns = 6},
]

\addplot coordinates{
(1,4)
(2,2)
(3,2)
(4,1)
(5,1)
(6,1)
(7,1)
};

\addplot coordinates{
(1,3.26554829897127)
(2,2.09979937678747)
(3,1.93204422247834)
(4,1.20595893627012)
(5,1.20595893627012)
(6,1.20595893627012)
(7,1.08473129295258)
};

\addplot coordinates{
(1,3.32079308591764)
(2,2.03457041179461)
(3,2.03457041179461)
(4,1.09405185561769)
(5,1.21098118962888)
(6,1.21098118962888)
(7,1.09405185561769)
};

\addplot coordinates{
(1,3.33333333333333)
(2,2.12121212121212)
(3,2.12121212121212)
(4,1.00000000000000)
(5,1.21212121212121)
(6,1.21212121212121)
(7,1.00000000000000)
};
\end{axis}
\end{tikzpicture}
\end{figure}

\begin{example} \label{Examp5}
Consider the degree sequence $\mathbf{d} = \left[ 4,\, 2,\, 2,\, 1,\, 1,\, 1,\, 1 \right]^\top$. There exist three networks with this attribute up to isomorphism, depicted on Figure~\ref{Fig5}.\footnote{~See the proof at \url{http://www.math.unm.edu/~loring/links/graph_s09/degreeSeq.pdf}.} 
It can be checked that $\varepsilon = 0.1$ is reasonable. Generalized degrees with this value are presented on Figure~\ref{Fig6}.

Degree provides the centrality ranking $1 \succ (2 \sim 3) \succ (4 \sim 5 \sim 6 \sim 7)$ in all networks. \\
For $(N,A)$ (Figure~\ref{Fig5a}), generalized degree results in $1 \succ 2 \succ 3 \succ (4 \sim 5 \sim 6) \succ 7$ as nodes $4$, $5$ and $6$ are symmetric, they are closer to the center $1$ than $7$, which is also the case with nodes $2$ and $3$. \\
For $(N,A')$ (Figure~\ref{Fig5b}), generalized degree gives the ranking $1 \succ (2 \sim 3) \succ (5 \sim 6) \succ (4 \sim 7)$ as nodes $2$ and $3$, $4$ and $7$, and $5$ and $6$ are symmetric, and the last pair is closer to the center than the middle pair. \\
For $(N,A'')$ (Figure~\ref{Fig5c}), generalized degree results in the ranking $1 \succ (2 \sim 3) \succ (5 \sim 6) \succ (4 \sim 7)$ as nodes $2$ and $3$, $4$ and $7$, and $5$ and $6$ are symmetric, and the last pair is connected to the center contrary to the middle pair.

Note that generalized degree is only a tie-breaking rule of degree but all changes can be justified. 
Centrality values may also have a meaning: $x_1(0.1)(N,A') > x_1(0.1)(N,A)$ as node $1$ can more easily communicate with node $7$ in the network $(N,A')$, for example. It is also remarkable that $x_2(0.1)(N,A) > x_2(0.1)(N,A')$ and $x_5(0.1)(N,A) < x_5(0.1)(N,A') < x_5(0.1)(N,A'')$ (nodes $5$ and $6$ are symmetric in all cases).
\end{example}

\section{An interpretation of the measure} \label{Sec4}

In this section we present the meaning of the proposed measure on the network.
Let $C \in \mathbb{R}^{n \times n}$ be a matrix such that $c_{ij} = a_{ij}$ for all $i \neq j$ and $c_{ii} = \mathfrak{d} - d_i$ for all $i = 1,2, \dots ,n$. $C$ is a modified adjacency matrix with equal row sums, its diagonal elements are nonnegative but at least one of them is zero.
In other words, $C = \mathfrak{d} I - L$ and $\mathbf{x}(\varepsilon) = \left[ I + \varepsilon \left( \mathfrak{d} I - C \right) \right]^{-1} \mathbf{d}$. Let introduce the notation
\[
\beta = \frac{\varepsilon}{1 + \varepsilon \mathfrak{d}}.
\]
It gives 
\[
\mathbf{x}(\varepsilon) = \frac{1}{1 + \varepsilon \mathfrak{d}} \left( I - \beta C \right)^{-1} \mathbf{d} = \left( 1 - \beta \mathfrak{d} \right) \left( I - \beta C \right)^{-1} \mathbf{d}.
\]
Matrix $\left( I - \beta C \right)^{-1}$ can be written as a limit of an infinite sequence according to its Neumann series \citep{Neumann1877} if all eigenvalues of $\beta C$ are in the interior of the unit circle \citep[p.~618]{Meyer2000}.

\begin{proposition} \label{Prop4}
Let $(N,A) \in \mathcal{N}^n$ be a network. Then
\[
\mathbf{x}(\varepsilon) = \left( 1 - \beta \mathfrak{d} \right) \sum_{k=0}^{\infty} \left( \beta C \right)^k \mathbf{d} = \left( 1 - \beta \mathfrak{d} \right) \left( \mathbf{d} + \beta C \mathbf{d} + \beta^2 C^2 \mathbf{d} + \beta^3 C^3 \mathbf{d} + \dots \right).
\]
\end{proposition}

\begin{proof}
According to the \emph{Ger\v sgorin theorem} \citep{Gersgorin1931}, all eigenvalues of $L$ lie within the closed interval $\left[ 0, 2 \mathfrak{d} \right]$, so eigenvalues of $\beta C$ are within the unit circle if $\beta < 1 / \mathfrak{d}$. It is guaranteed due to $\beta = \varepsilon / (1 + \varepsilon \mathfrak{d})$.
\end{proof}

Multiplier $\left( 1 - \beta \mathfrak{d} \right) > 0$ in the decomposition of $\mathbf{x}(\varepsilon)$ is irrelevant for the centrality ranking, it just provides that $\sum_{i=1}^n x_i(\varepsilon) = \sum_{i=1}^n d_i$.

\begin{lemma} \label{Lemma2}
Generalized degree centrality measure $\mathbf{x}(\varepsilon) = \lim_{k \to \infty} \mathbf{x}(\varepsilon)^{(k)}$ where
\[
\mathbf{x}(\varepsilon)^{(0)} = \left( 1 - \beta \mathfrak{d} \right) \mathbf{d},
\]
\[
\mathbf{x}(\varepsilon)^{(k)} = \mathbf{x}(\varepsilon)^{(k-1)} + \left( 1 - \beta \mathfrak{d} \right) \left( \beta C \right)^k \mathbf{d}, \quad k = 1,2, \dots \, .\footnote{~Superscript $(k)$ indicates the centrality vector obtained after the $k$th iteration step.}
\]
\end{lemma}

\begin{proof}
It is the immediate consequence of Proposition~\ref{Prop4}.
\end{proof}

Lemma~\ref{Lemma2} has an interpretation on the network. In the following description the multiplier $\left( 1 - \beta \mathfrak{d} \right)$ is disregarded for the sake of simplicity. Let $G'$ be a graph identical to the network except that $\mathfrak{d} - d_i$ loops are assigned for node $i$. In this way balancedness is achieved with the minimal number of loops, at least one node (with the maximal degree) has no loops.
Graph $G'$ is said to be the \emph{balanced network} of $G$. It is the same procedure as balancing a multigraph by loops in \citet[p.~1495]{Chebotarev2012} and in \citet{Csato2015a}, where $G'$ is called the balanced-graph and balanced comparison multigraph of the original graph, respectively.

Initially all nodes are endowed with an own estimation of centrality by their degree. In the first step, degree of nodes connected to the given one is taken into account: $C \mathbf{d}$ corresponds to the sum of degree of neighbours (including the nodes available on loops). Adding $\mathfrak{d} - d_i$ loops provides that the number of $1$-long paths from node $i$ is exactly $\mathfrak{d}$.\footnote{~In the limit it corresponds to the average degree of neighbours in $G'$ since $\lim_{\varepsilon \to \infty} \beta = 1 / \mathfrak{d}$.}
Then this aggregated degree of objects connected to the given one is added to the original estimation with a weight $\beta$, resulting in $\mathbf{d} + \beta C \mathbf{d}$.

In the $k$th step, the summarized degree of nodes available on all $k$-long paths (including loops) $C^k \mathbf{d}$ is added to the previous centrality, weighted by $\beta^k$ according to the length of the paths. This iteration converges to the generalized degree ranking due to Lemma~\ref{Lemma2}.

Example~\ref{Examp6} illustrates the decomposition.

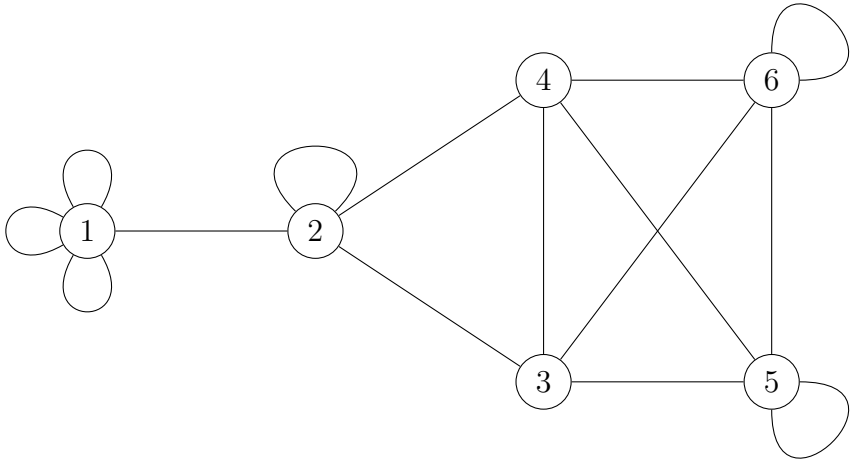
\begin{figure}[htbp]
\centering
\caption{Network of Example~\ref{Examp6}}
\label{Fig7}

\begin{tikzpicture}[scale=1,auto=center, transform shape, >=triangle 45]
\tikzstyle{every node}=[draw,shape=circle];
  \node (n1) at (0,0)  {$1$};
  \node (n2) at (3,0)  {$2$};
  \node (n3) at (6,-2) {$3$};
  \node (n4) at (6,2)  {$4$};
  \node (n5) at (9,-2) {$5$};
  \node (n6) at (9,2)  {$6$};

  \foreach \from/\to in {n1/n2,n2/n3,n2/n4,n3/n4,n3/n5,n3/n6,n4/n5,n4/n6,n5/n6}
    \draw (\from) -- (\to);

\path (n1) edge[in=60,out=120,looseness=8] (n1)
	  (n1) edge[in=150,out=210,looseness=8] (n1)
	  (n1) edge[in=240,out=300,looseness=8] (n1)
	  (n2) edge[in=45,out=135,looseness=8] (n2)
	  (n5) edge[in=0,out=270,looseness=8] (n5)
	  (n6) edge[in=0,out=90,looseness=8] (n6);
\end{tikzpicture}
\end{figure}

\begin{figure}[htbp]
\centering
\caption{Iterated degrees $\mathbf{d}^{(k)}$ in Example~\ref{Examp6}}
\label{Fig8}

\begin{tikzpicture}
\begin{axis}[width=0.8\textwidth, 
height=0.6\textwidth,
symbolic x coords={0,1,2,3,10,25,100},
xtick=data,
xlabel = Number of iterations ($k$),
ybar,
ymin = 0,
ymajorgrids = true,
bar width=8pt,
legend entries={Node $1\quad$,Node $2\quad$,Node $3$ and $4\quad$,Node $5$ and $6$},
legend style={at={(0.5,-0.2)},anchor = north,legend columns = 6}
]

\addplot coordinates{
(0,1.00000000000000)
(1,1.50000000000000)
(2,1.87500000000000)
(3,2.09375000000000)
(10,2.79648399353027)
(25,2.99152842760086)
};

\addplot coordinates{
(0,3.00000000000000)
(1,3.00000000000000)
(2,2.75000000000000)
(3,2.81250000000000)
(10,2.95195388793945)
(25,2.99800013303757)
};

\addplot coordinates{
(0,4.00000000000000)
(1,3.25000000000000)
(2,3.31250000000000)
(3,3.20312500000000)
(10,3.04804706573486)
(25,3.00199986696243)
};

\addplot coordinates{
(0,3.00000000000000)
(1,3.50000000000000)
(2,3.37500000000000)
(3,3.34375000000000)
(10,3.07773399353027)
(25,3.00323585271835)
};
\end{axis}
\end{tikzpicture}
\end{figure}
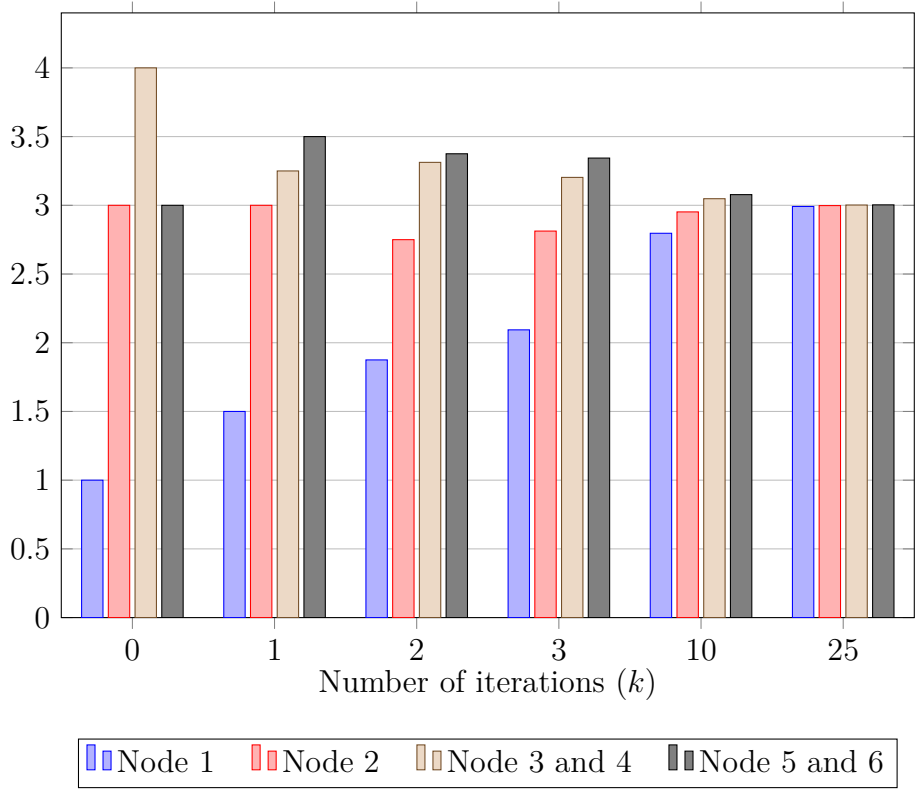

\begin{example} \label{Examp6}
Consider the network whose balanced network is shown on Figure~\ref{Fig7}, where the number of loops are determined by the differences $\mathfrak{d} - d_i$. Nodes $3$ and $4$, and $5$ and $6$ are symmetric, the two pairs have the same centrality for any $\varepsilon$. Generalized degree gives the rather natural ranking of $(3 \sim 4) \succ (5 \sim 6) \succ 2 \succ 1$.

Figure~\ref{Fig8} shows the average degree of neighbours available along a $k$-long path for various $k$-s, that is, $\mathbf{d}^{(k)} = \left[ \left( 1 / \mathfrak{d} \right) C \right]^{k} \mathbf{d}$. Their sum is equal to $\sum_{i=1}^n d_i$.
Lemma~\ref{Lemma2} means that, for instance, $\mathbf{x}(\varepsilon)^{(2)} = \left( 1 - \beta \mathfrak{d} \right) \left[ \mathbf{d}^{(0)} + \beta \mathfrak{d} \mathbf{d}^{(1)} + \beta^2 \mathfrak{d}^2 \mathbf{d}^{(2)} \right]$.
It reveals that nodes $5$ and $6$ are connected to more central nodes than $3$ and $4$. Another interesting fact is that $\mathbf{d}^{(k)}$ is monotonic only in its first coordinate. The elements of $\mathbf{d}^{(25)}$ are almost equal, large powers do not count much. Now $\mathbf{x}(\varepsilon)^{(1)}$ immediately gives the final ranking of nodes.
\end{example}

Two observations can be taken on the basis of examples discussed. The first is that ties in degree are usually eliminated after taking the network structure into account, which can be advantageous in practical applications: in the analysis of terrorist networks, a serious problem can be that standard centrality measures struggle to identify a given number of key thugs because of ties \citep{LindelaufHamersHusslage2013}. In other words, generalized degree has a good level of differentiation.

The second is the possibly slow convergence. In a sparse graph, long paths should be considered in order to get the final centrality ranking of the nodes, however, it is not clear why they still have some importance. Since the iteration depends on the network structure, it will be challenging to give an estimate for how many steps are necessary to approach the final centrality. 

\section{Conclusion} \label{Sec5}

The paper has introduced a new centrality measure called generalized degree. It is based on degree and uses the Laplacian matrix of the network graph. The method carries out a redistribution of a pool filled with the sum of degrees. The effect of neighbours centrality is controlled by a parameter, placing our method between degree and equal centrality for all nodes of a connected component. Inspired by the idea of \citet{Sabidussi1966}, a rank monotonicity axiom has been defined, and a sufficient condition has been provided in order to satisfy it. Besides PageRank \citep{Chienetal2004}, we do not know any other centrality measure with this property.
Furthermore, an iterative formula has been given for the calculation of generalized degree along with an interpretation on the network.

The main advantage of our measure is its degree-based concept. It is recommended to use with a low value of $\varepsilon$ instead of degree, which preserves most favourable properties of degree but has a much stronger ability to differentiate among the nodes and better reflect their role in the network (see Example~\ref{Examp6}). It is especially suitable to be a tie-breaking rule for degree, possible applications involve all fields where degree is used in order to measure centrality.

It is also suggested to test various parameters and follow all changes in the centrality ranking. The interval examined is not necessary to restrict to reasonable values, generalized degree may give an insight about the importance of nodes even if adding rank monotonicity is not guaranteed as its other properties remain valid. 

The research has opened some ways for future work. Generalized degree can be compared with other centrality measures, for example, through the investigation of their behaviour on randomly generated networks.
Rank monotonicity is worth to consider in an axiomatic comparison of centrality measures in the traces of \citet{LandherrFriedlHeidemann2010} and \citet{BoldiVigna2014}.
The iterative formula and the graph interpretation may also inspire a characterization of generalized degree.

Finally, some papers have used centrality measures just to describe the network by a single value of \emph{centrality index}. For instance, \cite{Sabidussi1966} has suggested that $1 / \max \{ f_i(N,A): i \in N \}$ is a good centrality index if centrality measure $f: \mathcal{N}^n \to \mathbb{R}^n$ satisfies some axioms.
Generalized degree improves on a failure of degree: $1 / \max \{ d_i: i \in N \}$ is the same in a complete and a star network of the same order but $1 / \max \{ x_i(\varepsilon): i \in N \}$ is larger in a complete one. We think this observation deserves more attention.

\section*{Appendix}

\emph{Proof of Theorem~\ref{Theo1}}.
Let $x_\ell = x_\ell(\varepsilon)(N,A)$ and $x_\ell' = x_\ell(\varepsilon)(N,A')$ for all $\ell \in N$. It can be assumed without loss of generality that $x_i'-x_i \geq x_j'-x_j$. Let $s = x_i'-x_i$, $t = x_j'-x_j$, $u = \max \lbrace x_\ell'-x_\ell: \ell \in N \setminus \{ i,j \} \rbrace = x_k' - x_k$ and $v = \min \{ x_\ell'-x_\ell: \ell \in N \setminus \{ i,j \} = x_m' - x_m$.
It will be verified that $t \geq u$.

Assume to the contrary that $t < u$. The difference of equations from \eqref{eq_main} for node $k$ gives:
\begin{eqnarray*}
\left( x_k' - x_k \right) + \varepsilon \sum_{\ell \in N \setminus \{ i,j,k \} } a_{k \ell} \left[ \left( x_k' - x_k \right) - \left( x_\ell' - x_\ell \right) \right] + & & \\
+ \varepsilon a_{k i} \left[ \left( x_k' - x_k \right) - \left( x_i' - x_i \right) \right] + \varepsilon a_{k j} \left[ \left( x_k' - x_k \right) - \left( x_j' - x_j \right) \right] & = & d_k'-d_k = 0.
\end{eqnarray*} 
Since $u = x_k' - x_k \geq x_\ell' - x_\ell$ for all $\ell \in N \setminus \{ i,j \}$ and $u \geq t$, we get
\begin{equation} \label{eq1}
u \leq \varepsilon a_{k i} (s-u) \quad \Leftrightarrow \quad u \leq \frac{\varepsilon a_{k i}}{1+\varepsilon a_{k i}} s.
\end{equation}

Now take the difference of equations from \eqref{eq_main} for node $i$:
\[
s + \varepsilon \sum_{\ell \in N \setminus \{ i,j \} } a_{i \ell} \left[ s - \left( x_\ell' - x_\ell \right) \right] + \varepsilon \left( x_i' - x_j' \right) = d_i'-d_i = 1.
\]
Since $u \geq x_\ell' - x_\ell$ for all $\ell \in N \setminus \{ i,j \}$, we get
\[
s + \varepsilon \sum_{\ell \in N \setminus \{ i,j \} } a_{i \ell} (s-u) + \varepsilon \left( x_i' - x_j' \right) = s + \varepsilon d_i (s-u) + \varepsilon \left( x_i' - x_j' \right) \leq 1.
\]
An upper bound for $u$ is known from \eqref{eq1}, thus
\begin{equation} \label{eq2}
1 + \varepsilon \left( x_j' - x_i' \right) \geq  s + \varepsilon d_i (s-u) \geq s + \varepsilon d_i \left( 1- \frac{\varepsilon a_{k i}}{1 + \varepsilon a_{k i}} \right) s = \frac{1 + \varepsilon a_{k i} + \varepsilon d_i}{1 + \varepsilon a_{k i}} s.
\end{equation}
By combining \eqref{eq1} and \eqref{eq2}:
\begin{equation} \label{eq3}
u \leq \frac{\varepsilon a_{k i}}{1+\varepsilon a_{k i}} \frac{1 + \varepsilon a_{k i}}{1 + \varepsilon a_{k i} + \varepsilon d_i} \left[ 1 + \varepsilon \left( x_i' - x_j' \right) \right] = \frac{\varepsilon a_{k i}}{1 + \varepsilon a_{k i} + \varepsilon d_i} \left[ 1 + \varepsilon \left( x_j' - x_i' \right) \right].
\end{equation}

Take the difference of equations from \eqref{eq_main} for node $m$:
\begin{eqnarray*}
\left( x_m' - x_m \right) + \varepsilon \sum_{\ell \in N \setminus \{ i,j \} } a_{m \ell} \left[ \left( x_m' - x_m \right) - \left( x_\ell' - x_\ell \right) \right] + & & \\
+ \varepsilon a_{m i} \left[ \left( x_m' - x_m \right) - \left( x_i' - x_i \right) \right] + \varepsilon a_{m j} \left[ \left( x_m' - x_m \right) - \left( x_j' - x_j \right) \right] & = & d_m'-d_m = 0.
\end{eqnarray*} 
Since $v = x_m' - x_m \leq x_\ell' - x_\ell$ for all $\ell \in N \setminus \{ i,j \}$ and $s \geq t$, we get
\begin{equation} \label{eq4}
v \geq \varepsilon a_{m i} (s-v) + \varepsilon a_{m j} (t-v) \quad \Rightarrow \quad v \geq \frac{\varepsilon \left( a_{m i} + a_{m j} \right)}{1+\varepsilon \left( a_{m i} + a_{m j} \right)} t.
\end{equation}

The difference of equations from \eqref{eq_main} for node $j$ results in:
\[
t + \varepsilon \sum_{\ell \in N \setminus \{ i,j \} } a_{j \ell} \left[ t - \left( x_\ell' - x_\ell \right) \right] + \varepsilon \left( x_j' - x_i' \right) = d_j'-d_j = 1.
\]
Since $v \leq x_\ell' - x_\ell$ for all $\ell \in N \setminus \{ i,j \}$, we get
\[
t + \varepsilon \sum_{\ell \in N \setminus \{ i,j \} } a_{j \ell} (t-v) + \varepsilon \left( x_j' - x_i' \right) = t + \varepsilon d_j (t-v) + \varepsilon \left( x_j' - x_i' \right) \geq 1.
\]
A lower bound for $v$ is known from \eqref{eq4}, thus
\begin{equation} \label{eq5}
1 + \varepsilon \left( x_i' - x_j' \right) \leq t + \varepsilon d_j \left[ 1- \frac{\varepsilon \left( a_{m i} + a_{m j} \right)}{1 + \varepsilon \left( a_{m i} + a_{m j} \right)} \right] t = \frac{1 + \varepsilon \left( a_{m i} + a_{m j} \right) + \varepsilon d_j}{1 + \varepsilon \left( a_{m i} + a_{m j} \right)} t.
\end{equation}

According to our assumption $t < u$, therefore from \eqref{eq3} and \eqref{eq5}:
\[
\frac{1 + \varepsilon \left( a_{m i} + a_{m j} \right)}{1 + \varepsilon \left( a_{m i} + a_{m j} \right) + \varepsilon d_j} \left[ 1 + \varepsilon \left( x_i' - x_j' \right) \right] < \frac{\varepsilon a_{k i}}{1 + \varepsilon a_{k i} + \varepsilon d_i} \left[ 1 + \varepsilon \left( x_j' - x_i' \right) \right].
\]

Obviously, it does not hold if $\varepsilon \to 0$. Now an upper bound is determined for the parameter $\varepsilon$. After some calculations we get:
\[
\frac{1 + \varepsilon \left( a_{m i} + a_{m j} + d_i \right) + \varepsilon^2 \left[ d_i \left( a_{m i} + a_{m j} \right) - a_{k i} d_j \right]}{1 + \varepsilon \left( a_{m i} + a_{m j} + 2 a_{k i} + d_i \right) + \varepsilon^2 \left[ \left( 2 a_{k i} + d_i \right) \left( a_{m i} + a_{m j} \right) + a_{k i} d_j \right]} < \varepsilon \left( x_j' - x_i' \right).
\]
Introduce the notation $\alpha = \varepsilon \left( a_{m i} + a_{m j} + d_i \right) + \varepsilon^2 \left[ d_i \left( a_{m i} + a_{m j} \right) - a_{k i} d_j \right]$ and $y = 1 + 2 \varepsilon a_{ki} + \varepsilon^2 \left[ 2 a_{k i} \left( a_{m i} + a_{m j} + d_j \right) \right] > 1$.
Then the fraction on the left-hand side can be written as $(1 + \alpha) / (y + \alpha)$. However, $(1 + \alpha) / (y + \alpha) > 1 / y$ because $y + \alpha > 0$ and $y > 1$, hence
\[
\frac{1}{1 + 2 \varepsilon a_{k i} + \varepsilon^2 \left[ 2 a_{k i} \left( a_{m i} + a_{m j} + d_j \right) \right]} < \varepsilon \left( x_j' - x_i' \right).
\]
Here $a_{k i},a_{m i},a_{m j} \leq 1$ and $d_j \leq \mathfrak{d}$. As $x_j' - x_i' \leq x_j - x_i \leq \mathfrak{d} - \mathsf{d}$ from boundedness (Proposition~\ref{Prop1}):
\[
1 < (\mathfrak{d} - \mathsf{d}) \left[ (2\mathfrak{d} + 4) \varepsilon^3 + 2 \varepsilon^2 + \varepsilon \right],
\]
which contradicts to the condition of Theorem~\ref{Theo1}.
\hfill{\ding{111}}


\begin{thebibliography}{}

\bibitem[Avrachenkov et~al., 2015]{AvrachenkovMazalovTsynguev2015}
Avrachenkov, K.~E., Mazalov, V.~V., and Tsynguev, B.~T. (2015).
\newblock Beta current flow centrality for weighted networks.
\newblock In Thai, M.~T., Nguyen, N.~P., and Shen, H., editors, {\em
  Computational Social Networks}, volume 9197 of {\em Lecture Notes in Computer
  Science}. Springer.

\bibitem[Bavelas, 1948]{Bavelas1948}
Bavelas, A. (1948).
\newblock A mathematical model for group structures.
\newblock {\em Human Organization}, 7(3):16--30.

\bibitem[Boldi and Vigna, 2014]{BoldiVigna2014}
Boldi, P. and Vigna, S. (2014).
\newblock Axioms for centrality.
\newblock {\em Internet Mathematics}, 10(3-4):222--262.

\bibitem[Bonacich, 1987]{Bonacich1987}
Bonacich, P. (1987).
\newblock Power and centrality: a family of measures.
\newblock {\em American Journal of Sociology}, 92(1):1170--1182.

\bibitem[Borgatti and Everett, 2006]{BorgattiEverett2006}
Borgatti, S.~P. and Everett, M.~G. (2006).
\newblock A graph-theoretic perspective on centrality.
\newblock {\em Social Networks}, 28(4):466--484.

\bibitem[Brin and Page, 1998]{BrinPage1998}
Brin, S. and Page, L. (1998).
\newblock The anatomy of a large-scale hypertextual web search engine.
\newblock {\em Computer networks and ISDN systems}, 30(1):107--117.

\bibitem[Chebotarev, 1989]{Chebotarev1989_eng}
Chebotarev, P. (1989).
\newblock Generalization of the row sum method for incomplete paired
  comparisons.
\newblock {\em Automation and Remote Control}, 50(8):1103--1113.

\bibitem[Chebotarev, 1994]{Chebotarev1994}
Chebotarev, P. (1994).
\newblock Aggregation of preferences by the generalized row sum method.
\newblock {\em Mathematical Social Sciences}, 27(3):293--320.

\bibitem[Chebotarev, 2012]{Chebotarev2012}
Chebotarev, P. (2012).
\newblock The walk distances in graphs.
\newblock {\em Discrete Applied Mathematics}, 160(10):1484--1500.

\bibitem[Chebotarev and Shamis, 1997]{ChebotarevShamis1997b_eng}
Chebotarev, P. and Shamis, E. (1997).
\newblock The matrix-forest theorem and measuring relations in small social
  groups.
\newblock {\em Automation and Remote Control}, 58(9):1505--1514.

\bibitem[Chien et~al., 2004]{Chienetal2004}
Chien, S., Dwork, C., Kumar, R., Simon, D.~R., and Sivakumar, D. (2004).
\newblock Link evolution: analysis and algorithms.
\newblock {\em Internet Mathematics}, 1(3):277--304.

\bibitem[Csat\'o, 2015]{Csato2015a}
Csat\'o, L. (2015).
\newblock A graph interpretation of the least squares ranking method.
\newblock {\em Social Choice and Welfare}, 44(1):51--69.

\bibitem[Dequiedt and Zenou, 2015]{DequiedtZenou2015}
Dequiedt, V. and Zenou, Y. (2015).
\newblock Local and consistent centrality measures in networks.
\newblock Manuscript. \url{https://www.gate.cnrs.fr/IMG/pdf/Dequiedt2015.pdf}.

\bibitem[Freeman, 1977]{Freeman1977}
Freeman, L.~C. (1977).
\newblock A set of measures of centrality based on betweenness.
\newblock {\em Sociometry}, 40(1):35--41.

\bibitem[Freeman, 1979]{Freeman1979}
Freeman, L.~C. (1979).
\newblock Centrality in social networks: conceptual clarification.
\newblock {\em Social Networks}, 1(3):215--239.

\bibitem[Garg, 2009]{Garg2009}
Garg, M. (2009).
\newblock Axiomatic foundations of centrality in networks.
\newblock Manuscript.
  \url{http://papers.ssrn.com/sol3/papers.cfm?abstract_id=1372441}.

\bibitem[Ger{\v s}gorin, 1931]{Gersgorin1931}
Ger{\v s}gorin, S. (1931).
\newblock {\"U}ber die {A}bgrenzung der {E}igenwerte einer {M}atrix.
\newblock {\em Bulletin de l'Acad\'emie des Sciences de l'URSS. Classe des
  sciences math\'ematiques et naturelles}, 6:749--754.

\bibitem[Gonz\'alez-Díaz et~al., 2014]{Gonzalez-DiazHendrickxLohmann2013}
Gonz\'alez-Díaz, J., Hendrickx, R., and Lohmann, E. (2014).
\newblock Paired comparisons analysis: an axiomatic approach to ranking
  methods.
\newblock {\em Social Choice and Welfare}, 42(1):139--169.

\bibitem[Jackson, 2010]{Jackson2010}
Jackson, M.~O. (2010).
\newblock {\em Social and economic networks}.
\newblock Princeton University Press.

\bibitem[Katz, 1953]{Katz1953}
Katz, L. (1953).
\newblock A new status index derived from sociometric analysis.
\newblock {\em Psychometrika}, 18(1):39--43.

\bibitem[Kitti, 2012]{Kitti2012}
Kitti, M. (2012).
\newblock Axioms for centrality scoring with principal eigenvectors.
\newblock Manuscript. \url{http://www.ace-economics.fi/kuvat/dp79.pdf}.

\bibitem[Klein, 2010]{Klein2010}
Klein, D.~J. (2010).
\newblock Centrality measure in graphs.
\newblock {\em Journal of Mathematical Chemistry}, 47(4):1209--1223.

\bibitem[Landherr et~al., 2010]{LandherrFriedlHeidemann2010}
Landherr, A., Friedl, B., and Heidemann, J. (2010).
\newblock A critical review of centrality measures in social networks.
\newblock {\em Business \& Information Systems Engineering}, 2(6):371--385.

\bibitem[Leavitt, 1951]{Leavitt1951}
Leavitt, H.~J. (1951).
\newblock Some effects of certain communication patterns on group performance.
\newblock {\em The Journal of Abnormal and Social Psychology}, 46(1):38.

\bibitem[Lindelauf et~al., 2013]{LindelaufHamersHusslage2013}
Lindelauf, R. H.~A., Hamers, H. J.~M., and Husslage, B. G.~M. (2013).
\newblock Cooperative game theoretic centrality analysis of terrorist networks:
  the cases of {J}emaah {I}slamiyah and {A}l {Q}aeda.
\newblock {\em European Journal of Operational Research}, 229(1):230--238.

\bibitem[Masuda et~al., 2009]{MasudaKawamuraKori2009}
Masuda, N., Kawamura, Y., and Kori, H. (2009).
\newblock Analysis of relative influence of nodes in directed networks.
\newblock {\em Physical Review E}, 80(4):046114.

\bibitem[Masuda and Kori, 2010]{MasudaKori2010}
Masuda, N. and Kori, H. (2010).
\newblock Dynamics-based centrality for directed networks.
\newblock {\em Physical Review E}, 82(5):056107.

\bibitem[Meyer, 2000]{Meyer2000}
Meyer, C.~D. (2000).
\newblock {\em Matrix analysis and applied linear algebra}.
\newblock Society for Industrial and Applied Mathematics, Philadelphia.

\bibitem[Mohar, 1991]{Mohar1991}
Mohar, B. (1991).
\newblock The {L}aplacian spectrum of graphs.
\newblock In Alavi, Y., Chartrand, G., Oellermann, O.~R., and Schwenk, A.~J.,
  editors, {\em Graph Theory, Combinatorics, and Applications}, volume~2, pages
  871--898. Wiley, New York.

\bibitem[Monsuur and Storcken, 2004]{MonsuurStorcken2004}
Monsuur, H. and Storcken, T. (2004).
\newblock Centers in connected undirected graphs: an axiomatic approach.
\newblock {\em Operations Research}, 52(1):54--64.

\bibitem[Neumann, 1877]{Neumann1877}
Neumann, C. (1877).
\newblock {\em Untersuchungen \"uber das logarithmische und {N}ewton'sche
  {P}otential}.
\newblock B. G. Teubner, Leipzig.

\bibitem[Nieminen, 1974]{Nieminen1974}
Nieminen, J. (1974).
\newblock On the centrality in a graph.
\newblock {\em Scandinavian Journal of Psychology}, 15(1):332--336.

\bibitem[Ranjan and Zhang, 2013]{RanjanZhang2013}
Ranjan, G. and Zhang, Z.-L. (2013).
\newblock Geometry of complex networks and topological centrality.
\newblock {\em Physica A: Statistical Mechanics and its Applications},
  392(17):3833--3845.

\bibitem[Rubinstein, 1980]{Rubinstein1980}
Rubinstein, A. (1980).
\newblock Ranking the participants in a tournament.
\newblock {\em SIAM Journal on Applied Mathematics}, 38(1):108--111.

\bibitem[Sabidussi, 1966]{Sabidussi1966}
Sabidussi, G. (1966).
\newblock The centrality index of a graph.
\newblock {\em Psychometrika}, 31(4):581--603.

\bibitem[Seeley, 1949]{Seeley1949}
Seeley, J.~R. (1949).
\newblock The net of reciprocal influence: a problem in treating sociometric
  data.
\newblock {\em Canadian Journal of Psychology}, 3(4):234--240.

\bibitem[Wasserman and Faust, 1994]{WassermanFaust1994}
Wasserman, S. and Faust, K. (1994).
\newblock {\em Social network analysis: methods and applications}.
\newblock Cambridge University Press.

\end{thebibliography}

\end{document}